\documentclass[preprint,12pt]{elsarticle}
\usepackage{lipsum}
\makeatletter
\def\ps@pprintTitle{%
	\let\@oddhead\@empty
	\let\@evenhead\@empty
	\def\@oddfoot{}%
	\let\@evenfoot\@oddfoot}
\makeatother

\usepackage{amssymb}

\usepackage{lineno}
\usepackage{enumitem}

\journal{}

\usepackage{rotating}

\usepackage{lmodern}
\usepackage{array}
\usepackage{graphicx}
\usepackage{multirow}
\usepackage{tabularx}
\usepackage{lscape}

\usepackage{rotating}
\usepackage{multirow}
\usepackage{fullpage}
\usepackage{anysize}
\marginsize{2cm}{2cm}{2cm}{2cm}
\usepackage{filecontents}
\usepackage{natbib}
\usepackage{bibentry}
\nobibliography*
\usepackage[boxed,ruled,vlined,linesnumbered]{algorithm2e}
\usepackage{url}
\usepackage{framed}
\usepackage{hyperref}
\usepackage{mdframed}

\usepackage{todonotes}
\usepackage{tikz}
\usepackage{tikz-cd}
\usetikzlibrary{matrix,shapes,arrows,positioning,chains, calc}
\usetikzlibrary{automata,matrix,shapes,arrows,positioning,chains,calc}
\usetikzlibrary{snakes}
\usetikzlibrary{arrows,scopes}
\usetikzlibrary{positioning,chains,fit,shapes,calc}
\usepackage{caption}
\usepackage{subcaption}
\usepackage[none]{hyphenat}
\usepackage{times}
\usepackage{amsmath}
\SetKwHangingKw{Local}{Local}
\SetKwHangingKw{Global}{Global}
\SetKwHangingKw{Constant}{Constant}
\SetKwHangingKw{Input}{Input}
\SetKwHangingKw{Alias}{Alias}
\SetKwHangingKw{Output}{Output}
\SetKwHangingKw{SideEffect}{Side effects}
\SetKwHangingKw{InternalEvent}{Internal event}
\SetKwHangingKw{External}{External}
\SetKwHangingKw{ExternalEvent}{External event}
\SetKwHangingKw{Precondition}{Precondition :}
\SetKwHangingKw{Postcondition}{Postcondition :}
\SetKwHangingKw{Upon}{Upon}
\SetKwHangingKw{DoForever}{Do Forever}
\SetKwHangingKw{PVab}{Persistent variables:}

\usepackage{theorem}
\usepackage{wrapfig}
\usepackage{amssymb}
\usepackage{xspace}
\newtheorem{assumption}{Assumption}
\newtheorem{condition1}{Condition}
\newtheorem{definition}{Definition}

\newtheorem{theorem}{Theorem}

\newenvironment{proof}[1][Proof]{\noindent\textbf{#1.} }{\ \rule{0.5em}{0.5em}}
\newtheorem{proposition}{Proposition}

\newtheorem{lemma}[theorem]{Lemma}

\newcommand{\B}{\vspace*{-\smallskipamount}}

\newcommand{\BBB}{\vspace*{-\bigskipamount}}

\newcommand{\remove}[1]{}

\usepackage[T1]{fontenc}
\usepackage[utf8]{inputenc}

\usepackage{setspace}

\usepackage{csquotes}
\usepackage{enumitem}
\usepackage{natbib}
\usepackage{tablefootnote}
\usepackage{longtable}
\usepackage{fnpct}

\usepackage{multirow}
\usepackage{hhline}

\begin{document}
\title{Peripheral Authentication for Parked Vehicles over Wireless Radio Communication\footnote{This is a full version of a short paper published in IEEE NCA 2016. In the cover letter, we provide details of the extended part in this version.}}
\begin{frontmatter}

\date{}

  \author[add1]{Shlomi Dolev\corref{cor1}}
  \ead{dolev@cs.bgu.ac.il}


  \author[add2]{Nisha Panwar}
  \ead{npanwar@uci.edu}

\cortext[cor1]{The first author's research was partially supported by the Rita Altura Trust Chair in Computer Sciences; the Lynne and William Frankel Center for Computer Science; the grant from the Ministry of Science, Technology and Space, Israel, and the National Science Council (NSC) of Taiwan; the Ministry of Foreign Affairs, Italy; the Ministry of Science, Technology and Space, Infrastructure Research in the Field of Advanced Computing and Cyber Security and the Israel National Cyber Bureau.}

  \address[add1]{Ben-Gurion University of the Negev, Israel.}
  \address[add2]{University of California, Irvine, USA.}

\begin{abstract}
Peripheral authentication is an important aspect in the vehicle networks to provide services to only authenticated peripherals and a security to internal vehicle modules such as anti-lock braking system, power-train control module, engine control unit, transmission control unit, and tire pressure monitoring. This paper presents a vehicle to a peripheral device and peripheral device to vehicle authentication scheme that verifies a binding between a vehicle and an authentic user peripheral device. In particular, a three-way handshake scheme is proposed for a vehicle to a keyfob authentication. A keyfob is a key with a secure hardware that communicates and authenticates the vehicle over the wireless channel. Usually, a secret pin number is entered through the wireless keyfob and the pin must be verified before receiving an access to the vehicle. Conventionally, a vehicle to keyfob authentication is realized through a challenge-response verification protocol. An authentic coupling between the vehicle identity and the keyfob avoids any illegal access to the vehicle. However, these authentication messages can be relayed by an active adversary, thereby, can amplify the actual distance between an authentic vehicle and a keyfob. Eventually, an adversary can possibly gain access to the vehicle by relaying wireless signals and without any effort to generate or decode the secret credentials. Hence, the vehicle to keyfob authentication scheme must contain an additional attribute verification such as physical movement of a keyfob holder.

Our solution is a two-party and three-way handshake scheme with proactive and reactive commitment verification. 
The proposed solution also uses a time interval verification such that both vehicle and keyfob would yield a similar locomotion pattern of a dynamic keyfob within a similar observational time interval. Hence, the solution is different from the distance bounding protocols that require multiple iterations for the round-trip delay measurement. The proposed scheme is shown to be adaptable with the existing commitment scheme such as Schnorr identification scheme and Pedersen commitment scheme.	
\end{abstract}

\end{frontmatter}


\pagenumbering{arabic}
\setcounter{page}{1}
\newpage
\section{Introduction}
\label{sect:Intro}
Currently, vehicles are leveraged as a secure mobile information system~\cite{model,secureit} in the Internet of Things (IoT) environment such as smart cities, smart communities, smart contracts, smart homes, etc. In order to allow the wireless communication capabilities, these vehicles must be compliant with the Dedicated Short Range Communication (DSRC) IEEE 1609~\cite{eded,od} based on Wireless Access in Vehicular Environment (WAVE) 802.11p~\cite{tuto}. There has been a tremendous amount of research on how to securely drive vehicles while simultaneously communicating with neighboring vehicles so that a warning can be received/predicted ahead of time. Examples for different types of communication among neighboring vehicles are vehicle platooning protocols for fuel efficiency, vehicle tracking protocols for traffic efficiency, data dissemination and offloading protocols for road side units (RSU), vehicle localization protocols for safety and infotainment routing, etc. However, another crucial aspect is to authorize an access to a vehicle via peripheral device connections~\cite{beath}. Our focus is to highlight the vulnerabilities attached to a static vehicle. It might be less intuitive to imagine the threat use cases for a {\em static} or parked vehicle. However, the static vehicle silhouette is even more vulnerable to all possible attacks once the vehicle is accessible to rogue peripheral devices. For example, an adversary can use repeaters to revive weak signals from a far distanced keyfob and thereby, receive an ahead of time access to the vehicle. Also, a brute-force method can be used to exhaustively compute the correct response and then forwarding it to the vehicle before the original response arrives at the vehicle. A more formal description regarding the attack scenarios is given in Section~\ref{sect:mod}. Therefore, we provide a general authentication for these peripheral devices and also improve the mobile keyfob authentication in a separate scheme. 

\medskip
\noindent\textit{Peripheral authentication:} A secure digital periphery of the vehicle is achieved via a secure authentication with respect to paired devices. Specifically, any temporary peripheral device connection with the vehicle must be authenticated for the extended functional security of the vehicle (see ISO 26262 vehicle functional security standard~\cite{iso}) while driving with a plugged-in rogue device. These peripheral devices such as keyfobs, USB sticks, cell phones, and, iPods provide extended services to the vehicle. Evidently, these ad hoc vehicle to device connections are potential exposure to external threats to break-in an otherwise static and secure vehicle periphery. Our motivation is to secure the peripheral device integration, especially, remote vehicle access via the keyfob. A keyfob ($\mathit{kf}$) is a hardware security token to allow only an authentic access to a static vehicle ($v$) situated remotely.

In particular, the problem is beyond the effort to place a secure firewall for filtering any external threats due to a range of relay and impersonation attacks (as presented in Section~\ref{sect:mod}). A secure remote access is most crucial among other peripheral device connections because vehicle access via a keyfob has a wider horizon of attacks. Therefore, it is important to identify, authenticate and pair the correct keyfob (while continuously approaching towards the parked vehicle) by measuring an active locomotion pattern of the keyfob and the keyfob holder. It must be noted that the proposed approach provides the event data history as a preliminary means to pair and authenticate any peripheral devices, however, in case of a keyfob an additional verification regarding keyfob dynamics is essential and provides stronger security against replay attacks.

\medskip

\noindent\textit{V2X communication paradigm:} There exist multiple dimensions to vehicular communication forte; for example, vehicle to infrastructure (V2I with road side units for service discovery across smart highways or smart cities), vehicle to cloud (V2C), vehicle to IoT (V2IoT communication with platoons or RFID enabled smart highways), vehicle to smart homes (V2SH), vehicle to peripheral device (V2PD for plug-in device authentication~\cite{surgon}), vehicle to vehicle (V2V) and vehicle to owner (V2O for personalized human-computer interaction). 

We chose to provide a secure authentication protocol for a vehicle to peripheral device (V2PD) communication. In general, our solution provides a secure binding between the vehicle and authentic peripheral devices. Furthermore, the solution is extended with an additional customization for vehicle to keyfob binding via user authentication. Basically, the solution considers a static vehicle and securely bind the available services at the vehicle via user authentication. It must be noted that a dynamic vehicle is vulnerable to a wider attack surface, yet a static vehicle offers the key to a variety of V2X communication paradigms (as mentioned above), hence it is more sensitive. Therefore, the protocol design requires a two-fold interactive authentication paradigm based on challenge-response verification along with the anthropomorphic features that include human aspects into verification. The internal vehicle networks are supposed to provide a secure identifying gateway to these external devices. However, every transient connection between the peripheral device and the vehicle must be verifiable. 

The proposed solution verifies the driver authenticity via a three-way handshake that promises a two-fold authentication. Furthermore, the handshake scheme derives all subsequent messages with an initial round of authentication associated to an initiator. A secure mutual pairing between the vehicle and peripheral devices~\cite{kumar} analogs ad hoc device pairing such as on Bluetooth. However, the proposed solution avoids any unauthorized access to a vehicle by using cryptographic commitment schemes as a building block. The proposed scheme avoids any consequent privileges to maliciously start the engine of a parked vehicle through peripheral device access. Our motivation is to strengthen
an access control over a static/parked vehicle such that firstly, an owner must be authenticated
based on pre-defined challenge-response pairs (CRP), secondly, verification of actively measured
dynamics as an attribution of owners' characteristics.

\begin{figure*}[ht]
	\begin{framed}
		\B
		\centering
		\includegraphics[scale=0.4]{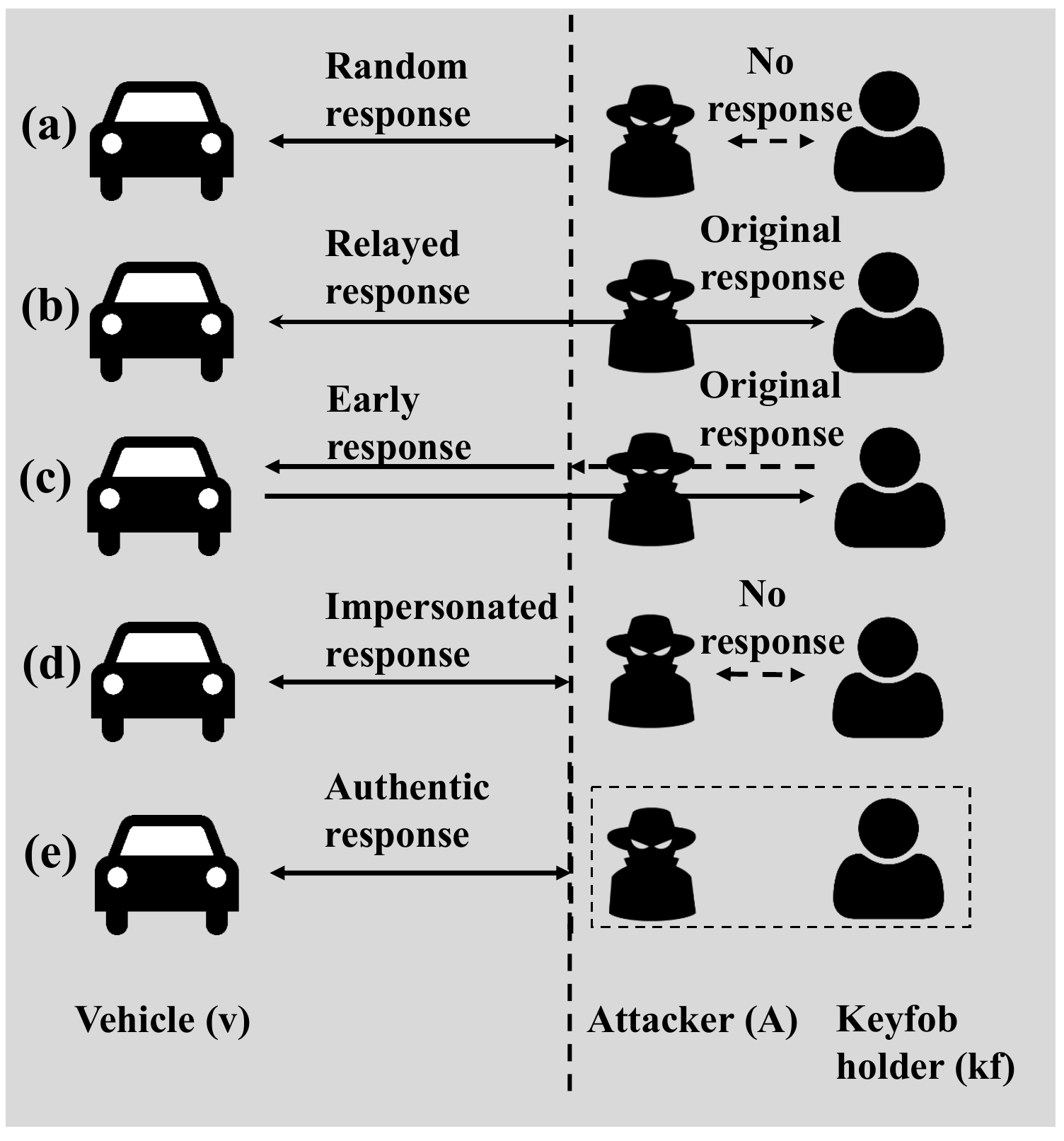}
		
		\raggedright
		
		In (a) adversary relays a random response via brute force attack. In (b) adversary purely relays an authentic response from a keyfob holder ($\mathit{kf}$) towards the vehicle ($v$). It must be noted that response was originated at authentic $\mathit{kf}$, meaning that an interested $\mathit{kf}$ has requested for service access at $v$. According to \textit{distance fraud} in (c), a dishonest prover $\mathit{kf}$ might fake a larger distance between the $\mathit{kf}$ and $v$, as if $\mathit{kf}$ is situated closer to $v$ by sending a response too early. However, according to \textit{mafia fraud} (as shown in (d)) $\mathit{kf}$ might not be interested in any service request (hence, might be situated outside a vehicle coverage) still an adversary have successfully revealed the secret from $\mathit{kf}$ and might control services on behalf of $\mathit{kf}$. It must be noted that original $\mathit{kf}$ is not interested and have not revealed the secret by colluding with the adversary. In (e) adversary colludes with the authentic $\mathit{kf}$, thereby, yield a service access via response originated at $\mathit{kf}$.

	\end{framed}
	\caption{Attack scenario.}
	\label{fig:yammy}
\end{figure*}

The digital periphery (meaning physical as well as wireless signal periphery) of a vehicle must utilize reactive and proactive commitment verification towards an access grantee. In general, a vehicle is supposed to receive a request for safe pairing and a subsequent access to various internal vehicle modules, e.g., anti-lock braking system (ABS), powertrain control module (PCM), engine control unit (ECU), transmission control unit (TCU), tire pressure monitoring (TPM), active control module (ACM), relay control module (RCM), heat ventilation and air condition (HVAC) systems. Evidently, the security of these modules is related to the secure pairing with peripheral devices. However, the secure pairing is even more crucial when a peripheral device (e.g., keyfob) requests a remote access to the vehicle.

\subsection{Problem Statement} Figure 1 illustrates a problem scenario where an attacker can get an access to a vehicle. A required solution must avoid an unauthorized remote vehicle access via fabricated radio frequency identification (RFID) enabled keyfob. Also, we focus on finding ways to provide an anthropomorphic link to the bonding between the
vehicle and peripheral devices such as RFID enabled keyfob. Conventionally, keyless entry systems provide an autonomous\footnote{Note that the system settings are defined within the scope of autonomous vehicles, i.e., availability of IEEE 802.11p~\cite{zerop}, IEEE 1609.2~\cite{eded}, and, Black Box IEEE 1616~\cite{motor}.} sensing such that the parked vehicle keep sensing (through heartbeat messages) the presence of an authentic keyfob in the proximity, e.g., via regular beacon solicitation method. The authentic keyfob must be present in the proximity and respond back to these soliciting beacons from the parked vehicle. However, the absence of the authentic keyfob within the sensed region can be amplified with another RFID enabled keyfob. The malicious keyfob would create an illusion of the shorter distance by amplifying and relaying the signals between both parties. In addition, these RFID signals are vulnerable to other sophisticated attacks as detailed in~\cite{kelo} while assuming an adaptive adversary model. In particular, an adversary recovers an exhaustive number of CRP transcripts and based on that knowledge might fabricate a duplicate keyfob. 






\subsection{Design Requirements} An authentication protocol construction must incorporate the verification of a pre-shared secret, an active response and a specific anthropomorphic feature. For example, the personalized locomotion pattern of a keyfob holder might learn an identifying information (more accurately with the passage of time). In particular, our design involves following factors and synergizes into a multi-dimensionally secure access control scheme. Essentially, design requirements can be summarized as:

\smallskip
\noindent\textit{Reciprocal authentication:} A primary requirement is to provide a mutual authentication between a vehicle and a peripheral device such as keyfob. In general, the vehicle to keyfob pairing is initiated from vehicle's side and keyfob as a responder. However, the vehicle as an initiator is more vulnerable to attack exposures as compared to the other way around. In our scheme, the keyfob is an initiator and the vehicle is a responder to validate a specific service access grantee such as an authentic keyfob. In this case, the vehicle authenticates the initiator and reciprocates a secret challenge along with the vehicle identity.

\smallskip
\noindent\textit{Identification based on pre-shared state:}
In our solution, an initial pairing is secured using an {\em internal state record} of the vehicle that is pro-actively synchronized with the recorded internal state inside the keyfob. The initial pairing must witness a matching internal state (inside the vehicle and an authentic keyfob) as a part of pre-shared knowledge verification phase.

\smallskip
\noindent\textit{Reactive verification:} The pro-active commitment (verification based on an internal vehicle state) must be coupled with a reactive commitment verification. This handshake ordering would avoid an attack scenario in which the adversary might respond with any random response to an authentic challenge. CRP based reactive verification avoids misbinding attacks and satisfies a \textit{non-injective} authentication property, i.e., to guarantee the participation of other party without the ability to distinguish it across multiple protocol executions. Thus, the vehicle must be able to verify the validity of the response, i.e., the response should be in correspondence with the current challenge.

\smallskip
\noindent\textit{Anthropomorphic features:} The personification of user traits (unique behavior or attribute of the keyfob owner) must be verified during the handshake. In particular, we need to verify velocity vs location with respect to the keyfob holder. Also, this locomotion pattern would become somewhat obvious and distinguishable over a period of time, i.e., any locomotion information collected over multiple authentication phases (between a specific vehicle and paired keyfob) would result into a personification of this locomotion pattern of the authentic owner. The human attribution of the owner's gait and active verification of the corresponding locomotion pattern is the classic form of authentication. In particular, this customization provides more intuitive authentication to an identical vehicle and keyfob over different sessions.

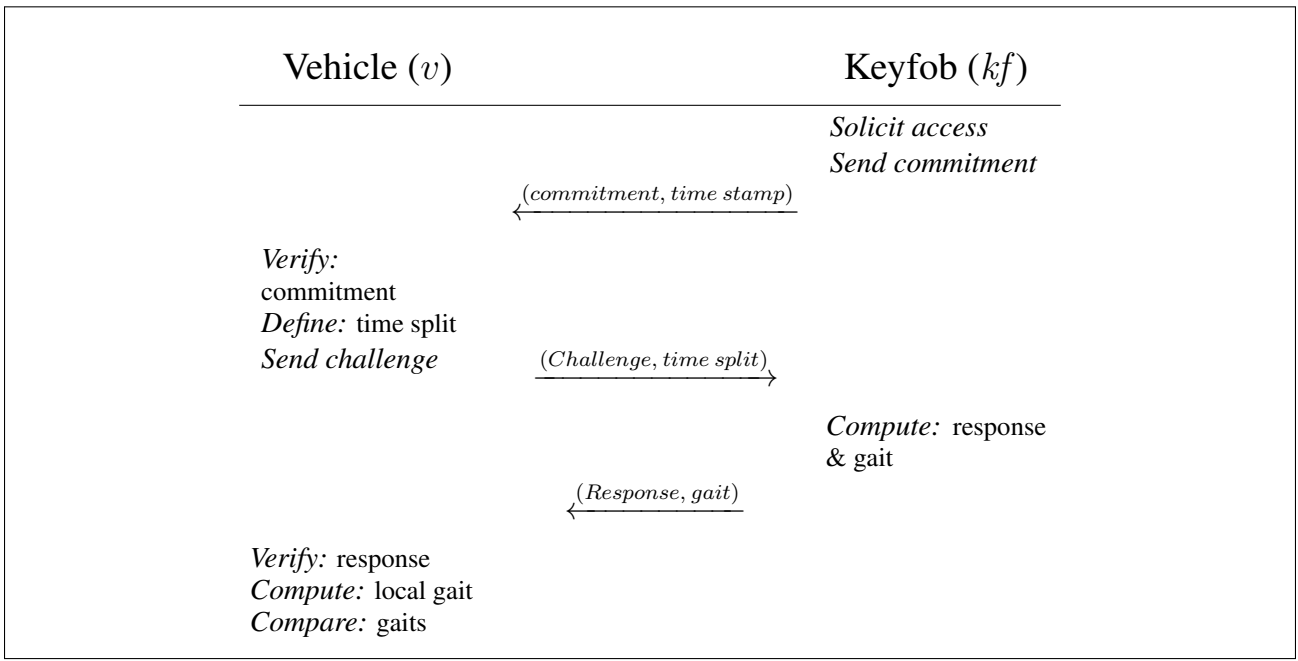
\begin{figure}[h]
	\captionsetup{belowskip=2pt,aboveskip=8pt}
		\begin{framed}
	\centering
	\begin{tikzpicture}
	\matrix (m)[matrix of nodes, column  sep=0.5mm,row  sep=1mm, nodes={draw=none, anchor=center,text depth=0pt} ]{
		\large{Vehicle ($v$)} & & \large{Keyfob ($\mathit{kf}$)} \\[1mm]
		\hline
		& & \begin{minipage}{2.8cm}
		\small\textit{Solicit access}
		
		\small\textit{Send commitment}
		\end{minipage} \\ 
		& $\xleftarrow{\small{(commitment, \: time \: stamp)}}$ & \\
		\begin{minipage}{2.8cm}
		\small\textit{Verify:} \footnotesize{commitment}
		
		\small\textit{Define:} \footnotesize{time split}
		
		\small\textit{Send challenge}
		\end{minipage} & \\ 
		& $\xrightarrow{\small{(Challenge, \: time \: split)}}$ & \\
		& &  \begin{minipage}{2.9cm}
		\small\textit{Compute:} \footnotesize{response \& gait}
		\end{minipage} & \\ 
		& $\xleftarrow{(\small{Response, \: gait})}$ & \\
		\begin{minipage}{3.1cm}
		\small\textit{Verify:} \footnotesize{response}
		
		\small\textit{Compute:} \footnotesize{local gait}
		
		\small\textit{Compare:} \footnotesize{gaits} \\
		
		\end{minipage} & \\ 
	};
	
	\end{tikzpicture}
		\end{framed}
	\caption{The proposed approach.}
	\BBB
	\label{fig:pap}
\end{figure}

%

\subsection{Our Contribution} Figure~\ref{fig:pap} presents a generalized vision of our proposed vehicle to device authentication. It must be noted that the first round is commitment revealing round, i.e., the residual commitment (that a vehicle and device were left with) during last communication is revealed and verified here at the first round of next connection establishment in near future. Thereby, a secure chain is created between the \textit{last} and the first round of connection establishment. In addition, the proposed solution is based on the following properties.


\smallskip\noindent\textit{Knowledge based initial pairing:} A most recent internal state, e.g., electronic control unit's (ECU) configuration, seat position, the angle of steering, temperature, etc, of any vehicle, is known only to the keyfob in current use. Therefore, every time a vehicle changes its internal state, the corresponding commitment data is also changed. In particular, only a key that was used during the last drive would be in possession of correctly matching internal state of the vehicle.

\smallskip\noindent\textit{Commitment based authentication:} A reactively changing secret commitment between the vehicle and authentic keyfob provides spontaneous identity verification. It provides the identity verification of the peer party in communication during the last rounds of the handshake. Furthermore, we have shown a secure adaptation of our proposed solution that utilizes existing commitment schemes, i.e., Schnorrs' identity-based commitment scheme~\cite{beryyinluv} and Pedersen's commitment scheme~\cite{chour} (as detailed in Section~\ref{sect:peddepi}).

\smallskip\noindent\textit{Personified localization:} An instant gait verification is done with respect to a keyfob holder approaching towards the paired vehicle. Evidently, a direct communication between the paired vehicle and corresponding keyfob would yield the same gait observation, i.e., similar distance covered in similar time window; as opposed to an attack scenario in which the adversary relay the wireless radio messages. In particular, no additional hardware integration is required within the internal vehicle network. The gait observation at the static vehicle would differ with the original gait observation at the keyfob holder if an active relaying is involved in the wireless radio communication.





\medskip\noindent\textbf{Outline.} Section~\ref{sect:mod} details the authentication model and attack scenarios regarding two-party authenticated communication. Section~\ref{sect:zio} provides system setting and assumptions. Furthermore, Section~\ref{sect:ano} provides the generalized peripheral authentication scheme as well as the personalized authentication scheme for vehicle to keyfob interaction. The ability to adapt potential commitment protocols from the existing literature is given in Section~\ref{sect:peddepi}. Furthermore, a formal security analysis is given in Section~\ref{sect:crpt}. Section~\ref{sect:prev} and Section~\ref{sect:fut} provide related work and concluding remarks, respectively.





\section{Authentication Properties and Threat Model}\label{sect:mod} In our model, we consider three entities: a verifier or a static vehicle ($v$), a prover or a keyfob ($\mathit{kf}$), and an attacker ($\mathcal{A}$). We consider a hierarchy of authentication properties such as aliveness, weak agreement, non-injective agreement, and agreement, to be a milestone for our proposed solution~\cite{hierarchy}. Our solutions, i.e., vehicle to peripheral device authentication and vehicle to keyfob authentication, satisfy authentication properties \textit{injective authentication} or \textit{agreement} as defined below.

\smallskip
\noindent \textit{Aliveness.} The most basic layer of authentication, i.e., aliveness property, ensures the participation of peer entities, while it might so happen that a participating entity was not running the protocol or even worse that it was running the protocol with some other entity.

\begin{definition}[Aliveness]
	A verifier $v$ verifies that a prover $\mathit{kf}$ has been participating in the execution of an authentication protocol with $v$.
\end{definition}

For example, a mirror attack scenario or an adversary using an old round of the protocol execution as an oracle with the initiator illustrates the lack of proper authentication over and beyond aliveness.

\smallskip
\noindent \textit{Weak agreement.} A step beyond aliveness is to guarantee the participation of an entity while the protocol execution actually took place between those claimed entities.

\begin{definition}[Weak agreement]
	A verifier $v$ verifies that a prover $\mathit{kf}$ has been participating in the execution of an authentication protocol with $v$, while protocol execution took place between $v$ and apparently from $\mathit{kf}$.
\end{definition}

\noindent \textit{Non-injective agreement.} This property of authentication ensures that each of the participants is present during the protocol execution and agrees on messages exchanged during the run. However, there might be a one-to-many association between the participants such that $v$ might be present during multiple executions while $\mathit{kf}$ was present in at most one of those executions.

\begin{definition}[Non-injective agreement]
	A verifier $v$ initiates protocol execution with a prover $\mathit{kf}$ and agrees on messages exchanged during the $n$ protocol executions, certainly, then prover $\mathit{kf}$ has been responding (hence agreeing on messages) in at least ($n-1$) protocol executions.
\end{definition}

\noindent \textit{Agreement.} One-to-one association between any two participants is termed as agreement or mutual authentication property. Thus, each of the participants agrees on their corresponding identity as well as the messages exchanged during \textit{each} protocol run.

\begin{definition}[Agreement]
	A verifier $v$ verifies that a prover $\mathit{kf}$ has been participating in the execution of an authentication protocol with $v$, while a mutual and affirmative protocol execution took place between $v$ and $\mathit{kf}$.
\end{definition}

Note that weak agreement and agreement properties are different; in the former, during the protocol run one party only guarantee the participation of other party; in later, one party guarantee the participation of other party and also preserves the distinctness corresponding to each protocol run. 

We present three different (but related) relay and impersonation attack scenarios, namely, distance fraud, mafia fraud, and terrorist fraud. These attack scenarios are applicable to various systems based on service access verification, especially, the services that are accessible over a wireless radio channel within a closer proximity. Therefore, as per the design requirements, the verifying vehicle $v$ must be able to guarantee: aliveness (i.e., the access request was actually generated at the original keyfob but does not avoid replay attack), weak agreement (i.e., the access request was actually generated at the authentic keyfob instead of an adversary replaying the past requests), non-injective agreement (i.e., the access request was actually generated at the authentic keyfob without preserving the distinctness (relay attack) across multiple executions of the protocol), agreement (i.e., the access request was actually generated at the authentic keyfob while preserving the distinctness and avoiding the replay attack).

\smallskip
\noindent\textbf{Distance fraud.} There are various services that are meant to verify the presence of a user in the locality before granting access to the resources. According to the distance fraud, a dishonest prover claims to be at a certain distance (thereby, legal to access the services) while actually being far away from the claimed distance with respect to the verifier. The dishonest prover pretends to fake a larger distance by sending a response too early to the verifier, i.e., even before the challenge reaching the prover and the prover computing a paired response and then sending a response to the verifier. For example, suppose that a vehicle and keyfob are situated apart (might be in the range of each other). An adversary might just amplify and relay these signals from the authentic keyfob pretending to acquire an access, or at worst, an adversary might send an early response to the vehicle (earlier than keyfob). In this case, an adversary would receive at least an early access, in case the original keyfob has actually requested for vehicle access. Therefore, both the proactive and reactive secret verification (with additional locomotion verification for the keyfob authentication) are crucial to the proposed authentication scheme. The authentication protocols that satisfy the only \textit{aliveness} property are usually prone to distance fraud. Therefore, the protocol must satisfy, at least, a \textit{weak agreement} property to avoid distance fraud attacks.

\smallskip
\noindent\textbf{Mafia fraud.} This attack is another more sophisticated form of distance attack. In this attack scenario, two adversaries collude and gain an illegitimate access to the secure services, while the original sender has no intention to request an access or to reveal a cryptographic secret. According to the mafia fraud, an adversary tries to utilize a separate channel and an accomplice to extract and relay the credentials from an authentic prover (not actively interested at all, in any service access). Therefore, adversaries collude and relay the secure access code to break in the verifier (meant to provide the service access to the authentic secret holder). The authentication protocols that satisfy only, \textit{aliveness} and \textit{weak agreement} properties are usually prone to mafia fraud. Therefore, the protocol must satisfy, at least, a \textit{non-injective agreement} property to avoid mafia fraud attacks.

\smallskip
\noindent\textbf{Terrorist fraud.} According to the terrorist fraud, an authentic prover assists with the adversary (by handing over a secret component) to impersonate in front of the verifier. The attack scenario is practically feasible even with the biometric authentication because the original secret holder can authorize himself, and, let the services be accessible to others who do not possess secret biometric credentials. A secure protocol design requires a more sophisticated identity verification method such as a ticket granting authority (issuing tickets for service access). It must be noticed that our solution is not resistant to this type of impersonation attacks. In general, authentication protocols that satisfy the most concrete form of authentication, i.e., \textit{injective agreement} property might still be prone to these attacks. A most resembling example is when any vehicle (that is secured under insurance policy of the owner) is stolen by a thief, if only, had the owner \textit{subliminally} assisted to thief.




Now we briefly enlist the solutions that might offer a secure communication link between a static vehicle and a dynamic keyfob at distance. However, none of these solutions incorporate the binding between the localization of a keyfob and the verification of authentic keyfob holder, as we do in our proposed approach. 


\medskip\medskip\noindent\textit{Strawman solutions:} In theory, non-transferrable machine-learning-fingerprints are also interesting, e.g., Completely Automated Public Turing Test to tell Computers and Humans Apart (CAPTCHA). More importantly, the CAPTCHA challenge should be in the form of unknown secret challenge while the response is examined by probabilistic estimations to decide the correct user identity binding with respect to that response. The CAPTCHA is based on reactive response verification to check the human capabilities of computing a correct response immediately. Interestingly, the reactive verification and anthropomorphic feature (e.g., walking pattern or gait observation, visual traits, typing patterns) can also be combined using CAPTCHA method of verification. The CAPTCHA provides identification of a person/device through a custom reaction to the CAPTCHA in order to avoid replay attacks, whereas the machine learning is trained to identify the person/device based on upcoming never seen CAPTCHAs. Conventionally, a new CAPTCHA is produced as a challenge for every user in a new interaction. This CAPTCHA challenge might invoke a user response not only in terms of the regular optical character recognition (OCR) but also in terms of counting the number of user movements per unit time, voice commands, or visual reflections, etc. All of these user-oriented commands qualify to be an anthropomorphic feature and can be used for verification purpose. However, there is no distinction between two different users from CAPTCHA perspective. Our protocol is able to distinguish two different users based on anthropomorphic attributes.

The existing solutions provide repetitive challenge-response iterations, triangulation and multi-alteration, and location mapping via RFID tag. Furthermore, the solution designs might vary as per the security features required and a few of those strawman solutions have been expanded as below:

\medskip\noindent\textit{Wireless fingerprinting:} The remote physical \textit{wireless device fingerprinting} as proposed in~\cite{beknown} is a useful solution to securely identify any remote device. According to the wireless device fingerprinting~\cite{baner}, a received signal is distinguished on the basis of four timing functions, i.e., ideal information encoded, unique manufacturing impairment in device hardware, environmental noise, and propagation delay.

$$\mathit{Received}(t)=\mathit{Ideal}(t)+\mathit{Impairment}(t)+\mathit{Noise}(t)+\mathit{Propogation}(t)$$

In~\cite{beknown} the remote device fingerprinting is based on \textit{clock skew} as observed through Transport Control Protocol (TCP) timestamps or Internet Control Message Protocol (ICMP) timestamps~\cite{skave} as embedded inside the header. The advantage is that their clock skew measurement can even be used with the Network Address Translation (NAT) protected networks which provides masking a single Internet Protocol (IP) address for multiple unique IP hosts. 

$$\mathit{Received}(t)=\mathit{Ideal}(t)+\mathit{Anthropomorphic}(t)$$

In addition, wireless channel and device fingerprinting can be used together to eliminate \textit{repeater-in-the-middle} attacks; consequently, nullify any artificial delays during the communication. Usually, in order to distort these fingerprinting signature, an attacker performs tampering such that $\mathit{Impairment}(t)\approx \mathit{Noise}(t)$. On the contrary, our proposed solution is different such that the received signal is distinguished based on two timing functions, i.e., ideal encoded information and user-oriented anthropomorphic features.

\medskip\noindent\textit{Measuring movement pattern:} 
Next, we discuss localization methods to securely identify a moving keyfob. The verifiable multi-lateration location technique~\cite{secpos} proposed a secure positioning scheme. Accordingly, a verifier estimates mobile node position through the distance bounding approach; similarly, multiple verifiers located in a triangle around the node can estimate the actual node position. We suggest using a reflection from passive ``verifiers'' (possibly, reflecting devices located at home or road-side lamps) for a proactive verification of any static vehicle. A trusted third party retrieves the position of these verifiers and their corresponding estimates (regarding mobile node) to accurately compute the mobile node location through \textit{minimum mean square} estimate~\cite{secpos}. However, our proposed solution would work even in the absence of these verifiers (that might not be trustworthy at times). In addition, we avoid the method of {\em multiple} measurements of round trip time (as in case of distance bounding) to locate the sender within the proximity.

The major difference in our proposed solution as compared with other solutions is that both a proactive and reactive commitment is verified while referring to similar time windows. The time window estimation (simultaneously, at both parties in communication) results in locomotion trajectory of the keyfob holder. Our solution might further be extended for computing multiple triangulation (for 2d positioning, 3rd positioning, and in fact the $i$'th) measurements over a specified period of time, thereby, locating the mobile node at different time stamps. These located dotes over the trajectory of coordinates will yield the direction of movement, velocity, and acceleration of the node. Therefore, a verifier can identify a node, by comparing its past movement pattern, say, using machine learning, and may predict the next location of the node as well (e.g., for avoiding collisions).

\section{System Settings and Protocol Overview}\label{sect:zio} This section presents system and hardware assumptions. Firstly, the vehicles must be compliant with the dedicated short range communication (DSRC) IEEE 1609 based on wireless access in vehicular environment (WAVE) 802.11p. In addition, a wake-up mode is required to allow a far distanced keyfob to reach the static vehicle in sleep mode. Also, a non-volatile storage medium is required that would securely accumulate the most recent driving patterns and the vehicle configurations. Next, a detailed working description of IEEE 1616 based data recorders is presented that serves the proactive knowledge verification between static a vehicle and a keyfob. Also, a brief overview of the proposed solution along with the participating entities and their pre-stored knowledge is highlighted.

\subsection{IEEE 1616 event data recorders (EDRs):} EDRs are used to maintain event primitives in a log. These essential event factors ($e$) contribute to improve consequent safety events~\cite{motor,mvdr} in future. Subsequently, any forensic investigations against a crash event would extract and link these event records about the vehicle. However, the privileges regarding the access to this critical/confidential information of any vehicle depend on the state law. It may further be of utmost importance to forensic team investigation. In past, consumers were not aware that an EDR is integrated inside the vehicle. However, the current state laws, i.e., Black Box Privacy Protection act 2013~\cite{act,meets}, made the EDR ownership clear to a vehicle owner and that it cannot be accessed without the consent from the vehicle owner on which it is integrated into (1) the presence and location of an EDR (also termed as a black box), (2) refining the critical event information and storage format, and (3) usage and claims to acquire the recorded internal state data for legal proceedings with owner's consent. We emphasize that the crucial event record stored on a volatile memory inside ECUs can be used to authorize the original vehicle owner of the vehicle.

\begin{definition}[EDR mobility pattern]
	\textit{The last most recent time frame of the vehicle dynamics during the last itinerary available from a non-volatile EDR storage defines the vehicle mobility pattern. A static vehicle and the corresponding keyfob share this mobility trace as a common internal state of the vehicle.}
\end{definition}

\smallskip
\noindent\textit{Record retrieval process:} IEEE 1616a~\cite{motor,mvdr} defines event observation and recording process. The event primitives such as acceleration, deceleration, steering angle/movement, velocity, and seat position, accounts for driving decisions taken during a crash incident. The record extraction is feasible as one of the connections given below:

\begin{itemize}
	\item A serial data link communication path to the {\em vehicle diagnostic communication port} that is on-board diagnostic (OBD-II) or SAE J1962~\cite{portouni}.
	
	\item A serial data link communication path connected via a direct cable to the target ECU (a partial set of events).
	
	\item A direct cable to the EEPROM component on the printed circuit board (PCB) within the EDR assembly. This method is more vulnerable to security attacks because it overrides the security barriers for module access control.
\end{itemize}


SAE J1962 port defines the requirements of an OBD connector and is technically equivalent to ISO/DIS 15031-3:2001. While motor vehicle event data recorder connector lockout apparatus (MVEDRCLA) is an amendment to IEEE 1616. It defines a device configuration to physically secure OBD II connector, primarily used to access the EDR (and other vehicle systems) such that it prevents tampering and unauthorized access to critical data. Surprisingly, there are only a few security integration protocols~\cite{meets} to enable data confidentiality, integrity, and availability regarding EDR components. EDR is deployed as a black box component whereas the input-output data security is left on to the vehicle vendor. A secure EDR integration to the internal vehicle network is assumed to be part of security features used for internal network components itself.


\subsection{Overview}
Our scheme provides a solution against the distance and mafia fraud attacks. Moreover, the terrorist fraud scenario requires additional assumptions, i.e., biometrics and consumer-friendly hardware to verify biometric features. A pre-processing phase and the subsequent usage of cryptographic primitives is given as below:

\smallskip
\noindent \textit{Setup phase (Key generation):} The manufacturing authority initializes a security module, $\mathit{Init(\mathit{Auth})}$ with a secret symmetric key $K$ (for both vehicle and keyfob) before even handing it over to the consumer.

\smallskip
\noindent \textit{Registration phase (Binding a symmetric key with the user identity):} The next phase is to bind a vehicle to keyfob at the time of handing it (a vehicle and a keyfob) over to a specific consumer $u$. Accordingly, the registration phase is required to associate the pre-initialized symmetric key, $K$, in the setup phase, $\mathit{Init(\mathit{Auth})}$, with the keyfob and the user, i.e., $\mathit(K,\mathit{kf},v,u)$: binding a user $u$ with a vehicle $v$, a keyfob $\mathit{kf}$, and the symmetric key $K$. It must be noticed that the user identity $u$ is crucial for the initial binding such as creating an administrator account. Initially, event records are null and do not provide a linkage between any keyfob holder, as in the past and in the current.


\smallskip
\noindent \textit{Query phase (Attempt to attack):} In the query phase, an adversary $\mathcal{A}$ utilizes the knowledge of the symmetric key $K$ (from the $n$ number of transcripts, say $t_n$, extracted during $n$ sessions in past) and performs the following sequence of message exchange:

\begin{itemize}
	\item The prover keyfob $\mathit{kf}$ sends $\mathit(K,\mathit{kf},v,u)$ to a verifier vehicle $v$ requesting for an access permission.
	
	\item An adversary $\mathcal{A}$ retrieves $(K,\mathit{kf},v,u)$ and relays $(K,\mathit{kf}_{\mathit{adv}},v,u)$ to the verifying vehicle $v$.
	
	\item The verifying vehicle $v$ verifies $(K,\mathit{kf}_{\mathit{adv}},v,u)$ before granting any access permission, i.e., $\mathit{Check} \leftarrow (K,\mathit{kf}_{\mathit{adv}},v,u)$.
\end{itemize}

Next, we present an authentication game to define the adversary advantage over the proposed scheme $\mathit{Auth}$. The security game for the proposed authentication scheme $\mathit{Auth}$, adversary $\mathcal{A}$, prover $\mathit{kf}$ and an authentic verifier $v$ is given as below:

\begin{proposition} An adversary $\mathcal{A}$ wins the game if $\mathit{[Check=1]}$. The probabilistic advantage of adversary, $\mathit{Adv(\mathcal{A})}$, for winning the game is
\end{proposition}

$$\mathit{Adv(\mathcal{A})}=Pr[\mathit{Check}=1]$$

Specifically, the vehicle to keyfob authentication scheme $\mathit{Auth}$ is secure if the $\mathit{Adv(\mathcal{A})}$ is negligible. In addition, $Pr[\mathit{Check}=1]$ is maximum during the initial rounds of pairing when event log is almost null. 



\begin{figure}[h!]
	\begin{center}
		\includegraphics[scale=0.4]{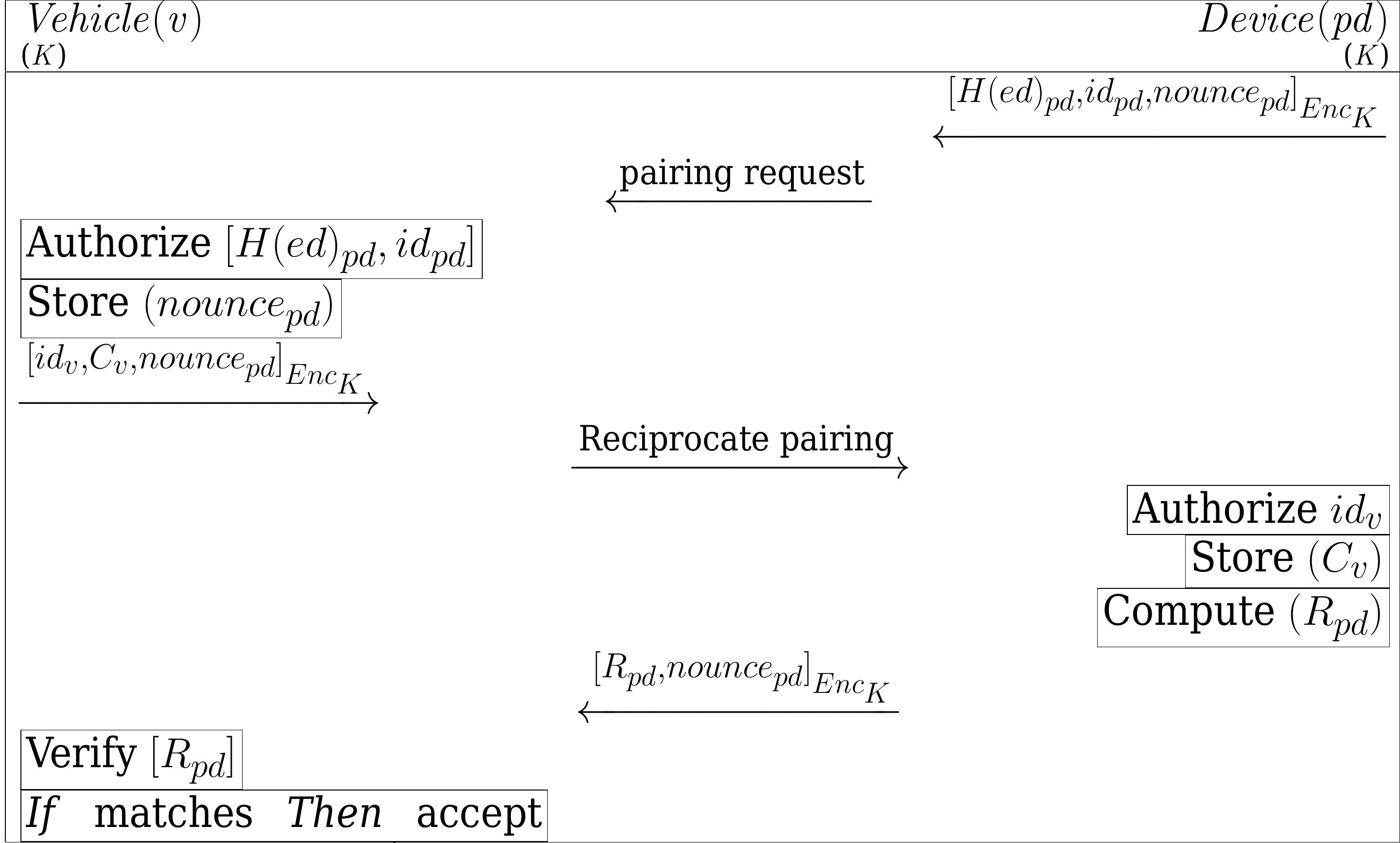}
	\end{center}
	\vspace{0.2cm}
	\BBB
	\caption{Peripheral authentication scheme.}
	\label{fig:alia}
\end{figure}

\section{Proposed Scheme}\label{sect:ano} We propose an authentication scheme (Figure~\ref{fig:alia}) for a remote vehicle access control that would authenticate the original keyfob of the vehicle and then give access to the paired vehicle. This access control can be perceived as an authorization check for services within the periphery of a vehicle. The authentication begins with an initial verification based on the event log replication. The event log is overwritten with every new event occurrence and keeps recording only recent few seconds of vehicle dynamics. Concurrently, vehicle event logs can be replicated over the peripheral device. Therefore, a peripheral device in possession of the most recent event logs would be authorized as a most recent occupant.

First, we present a simple authentication approach for peripheral devices, e.g., cell phone, USB stick, iPod, laptop, and other Bluetooth devices. Later, we propose a customized authentication scheme (in Figure~\ref{fig:gen}) based on the simple authentication scheme (as in Figure~\ref{fig:alia}).

\subsection{Peripheral Authentication Scheme: Basic Scheme}
As detailed in Figure~\ref{fig:alia}, a vehicle receives the pairing request from a peripheral device. The request begins with an encrypted identification information from a specific peripheral device $pd$ that sends the hash of shared internal state $H(ed)$, identity $id_{pd}$, and a sequence number $\mathit{nounce}_{pd}$. The event information is securely hashed in an abstract form. The vehicle and device share a symmetric key $K$ (as detailed in setup and registration phase), and all handshake messages are encrypted with $K$. Next, the vehicle verifies the pairing request as: (\textit{i}) Is the requesting device storing the similar internal state as $H(ed)_{v}=H(ed)_{pd}$ (i.e., the device has been used during most recent drive in past)? (\textit{ii}) Is the requesting device possessing an authentic identity as $pd_{id}$? (\textit{iii}) Is the requesting device using a unique nounce as $\mathit{nounce}_{pd}$? In the second step, after the successful verification of the initial commitment, the vehicle sends a challenge $C_v$ with the vehicle identity $id_{v}$, and the $\mathit{nounce}_{pd}$, all are encrypted using the key $K$. The device authenticates the reciprocated values, i.e., $id_v$, and paired $\mathit{nounce}_{pd}$. In addition, the device produces and sends the corresponding response $R_{pd}$ with $\mathit{nounce}_{pd}$. The device inputs the received challenge $C_v$ to the hash function $H$ and generates an output response as $H(C_v)$=$R_{pd}$. The verifying vehicle checks the reactive response $R_{pd}$, if the response match is found, then the vehicle accepts the device pairing request and grants the access.     

This authentication scheme might be vulnerable to distance attacks in case of a remote peripheral access, e.g., a remote keyfob to vehicle pairing. Therefore, we now present a customized scheme for vehicle to keyfob authentication.


%
%


\subsection{Vehicle to Keyfob Authentication Scheme: Customized Scheme}
The authentication scheme for a remote vehicle access is more crucial than a local peripheral device access. In our approach, (as detailed in Figure~\ref{fig:gen}), we provide a three-way handshake (similar to the basic scheme) such that each handshake is used to couple an earlier verified handshake. Also, the vehicle to keyfob authentication incorporates an authentication with human attribution. Next, we provide definitions for the personified attributes such as gait, displacement, and gait observation.


\begin{definition}[Gait]
	\textit{A dynamic keyfob holder has a distinguishable gait in terms of the velocity (displacement per unit time) and stride (distance covered per step).}
\end{definition}

\begin{definition}[Displacement]
	\textit{The distance covered by the keyfob holder at a varying velocity during the observation interval, i.e., from access solicitation (in the first step) to time split received (in the second step).}
\end{definition}

\begin{definition}[Gait observation]
	\textit{A dynamic keyfob holder, approaching the static vehicle and simultaneously requesting the access permission, should be distinguishable from any malicious third party. The vehicle verifies the deviation in local observation and original keyfob observation regarding dynamic locomotion pattern.}
\end{definition}

\begin{figure}[h]
	\begin{center}
		\includegraphics[scale=0.3]{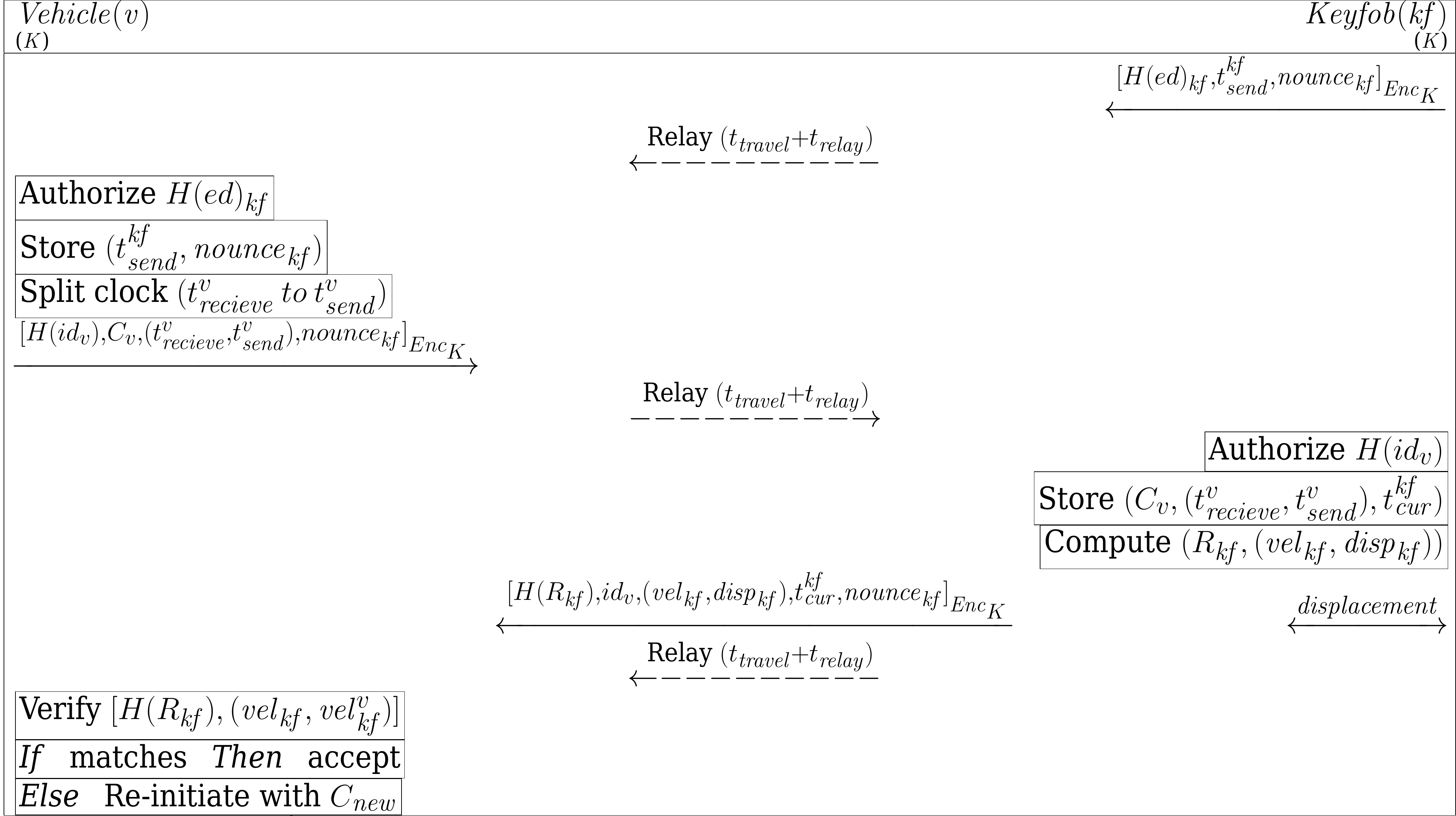}
	\end{center}
	\vspace{0.2cm}
	\BBB
	\caption{Vehicle to keyfob authentication.}
	\label{fig:gen}
\end{figure}

These attributes as mentioned above define an interactive verification of the owner. The three-way handshake as presented in Figure~\ref{fig:gen} is detailed below:

\begin{itemize}
	\item\textit{Keyfob soliciting the access permission:} In the {\em first step}, a keyfob holder initiates an access request to a static vehicle. The request should be encrypted with the symmetric key $K$ as originated from the authentic keyfob, accompanied with the EDR mobility pattern that was recorded during the last drive. The request ${[H(ed)_{\mathit{kf}},t^{\mathit{kf}}_{\mathit{send}},nounce_{\mathit{kf}}]}_{Enc_{K}}$ consists of hashed internal state $H(ed)$, which contains EDR mobility pattern, a time stamp $t^{\mathit{kf}}_{\mathit{send}}$, and a unique $nounce_{\mathit{kf}}$, for observing the gait and avoiding replay attacks, respectively.
	
	\item\textit{Vehicle challenging the keyfob:} In the {\em second step}, the vehicle verifies the received hashed EDR mobility pattern $H(ed)_{\mathit{kf}}$ with $H(ed)_{v}$ that must be known to an authentic keyfob that has the knowledge of the symmetric key $K$. Moreover, the vehicle mobility pattern must be known to a keyfob iff it was used during the last drive. The vehicle stores the time stamp $t^{\mathit{kf}}_{\mathit{send}}$ as the time when the access request was initialized by the keyfob holder. In addition, the vehicle stores a local time stamp $t^v_{\mathit{receive}}$ and the time difference ($t^v_{\mathit{receive}}-t^{\mathit{kf}}_{\mathit{send}}$), which will yield the propagation time that request has taken while traversing over the channel. It must be noticed that the presence of a relay attack in between the static vehicle and the dynamic keyfob holder will yield a higher propagation time as compared to regular communication, i.e., $(t^v_{\mathit{receive}}-t^{\mathit{kf}}_{\mathit{send}})_{\mathcal{A}}>(t^v_{\mathit{receive}}-t^{\mathit{kf}}_{\mathit{send}})_{regular}$, due to additional relaying times consumed at the adversray ($\mathcal{A}$) in middle, as compared to the regular communication. The vehicle sends an encrypted message ${[H(id_{v}),C_v,(t^v_{\mathit{receive}},t^v_{\mathit{send}}),nounce_{\mathit{kf}}]}_{Enc_K}$ containing a challenge $C_v$ and the local time stamps $t^v_{receive}$ to complete the gait observation of requesting keyfob.

	\item\textit{Authentication process at the keyfob:} Next, keyfob holder authorizes the signed vehicle id $H(id_{v})$ with the same nounce as was initialized in the first step (as shown in Figure~\ref{fig:gen}). The clock split $(t^v_{\mathit{receive}},t^v_{\mathit{send}})$ is used to compute the relative propagation time\footnote{For the practical deployament scenario time precision is a crucial issue, e.g., how frequently hardware clocks must be synchronized? what is the allowed measure for the clock drift? Therefore, we consider the strict event state verification along with the gait verification over the current clock-split window allowing a pre-defined amount of clock drift due to hardware limitations.} in each step. We assume that the generation of paired response $R_{\mathit{kf}}$ for the challenge $C_{v}$ in the current round is based on a keyed pseudo-random number generator that is not known to an adversray. Therefore, the adversary can only generate a random response $R^{\prime}$ (instead of actual $R_{kf}$) for the $C_{v}$, such that $R^{\prime} \neq R_{\mathit{kf}}$ without knowing the secret key $K$. Otherwise, an adversary could send a random $R^{\prime}$ ahead of time and demonstrate that $R^{\prime}$ is the authentic response (due to a lower propagation time from an adversary situated closer to the vehicle then the keyfob), i.e., $C_{v}$ traveled from vehicle to the original keyfob and an adversary in the middle responding with $R^{\prime}$ instead of $R_{kf}$ (which would have actually arrived a few seconds later then $R^{\prime}$).

	The keyfob will process the local gait observation for the time slot ($t^{\mathit{kf}}_{\mathit{cur}}-t^{\mathit{kf}}_{\mathit{send}}$). Therefore, the keyfob holder's displacement $\mathit{disp}_{\mathit{kf}}$ during this time interval will yield a velocity or a local gait observation $\mathit{vel}_{\mathit{kf}}$. Ideally, this gait observation at the keyfob holder and the same gait observation $\mathit{vel}^v_{\mathit{kf}}$ at the static vehicle must at least differ by an amount $t_{\mathit{relay}}$ to detect the presence of any relay attack.
	
	
	
	We would like to remark here that the two steps of handshake scheme verify proactive and reactive commitments. The last step of handshake scheme deals with the verification of the personalized properties of an access grantee. In particular, an anthropomorphic feature, e.g., walking pattern or gait observation, is verified. There can be different attributes to personify and distinguish the original owner, e.g., key typing speed pattern, face scan~\cite{sib}, shake pattern~\cite{well} etc. Our focus is to integrate those personified attributes into last round of handshake verification.
	
	\item\textit{Vehicle verifying the keyfob's gait:} 
	The static vehicle verifies the paired response $H(R_{\mathit{kf}})$ and the current session nounce. Next, the vehicle computes a local estimation of the keyfob holder's gait. Accordingly, the keyfob holder should have been displaced by $\mathit{disp}_{\mathit{kf}}$ (from the second step) over the time interval ($t^v_{\mathit{send}}-t^v_{\mathit{receive}}$).
	
	Now, considering the presence of a relay attacker in first and second steps, the local gait observation at the keyfob holder should differ by the local gait estimation at the static vehicle. Moreover, since the time window for the gait observation includes two propagation rounds; hence $2 \times (t_{\mathit{travel}}+t_{\mathit{relay}})$ should be suspiciously large enough with respect to $2 \times (t_{\mathit{travel}})$ in a regular communication. Furthermore, our approach can be extended to the well-known distance bounding protocols in which $n$ number of CRP exchanging rounds are required to compute a precise measure over the round trip propagation time. In our case, if the gait based measurement at the static vehicle differs with the local observation at the keyfob holder, then the protocol should be re-initialized with a new challenge $C_{\mathit{new}}$. Furthermore, even in a re-initialized session, the first scenario can be used as a base case in distance bounding protocols. Therefore, our solution can further be extended to provide a distance bounding protocol estimation.
\end{itemize}

Considering the issue of multiple drivers having access to the same vehicle during different time intervals; we suggest to incorporate an admin account (with full vehicle access) and a limited additional user accounts with partial/customized vehicle access. Therefore, each individual user account is accustomed to the correspondingly synchronized coupling between the keyfob and the internal state of the vehicle, i.e., $\mathit{keyfob}_{\mathit{admin}}$ is synchronized with the user $\mathit{account}_{\mathit{admin}}$. Similarly, a limited number of $n$ user accounts ${\mathit{user\:account_{[1\dots n]}}}$ are correspondingly synchronized with the ${\mathit{keyfob_{[1\ldots n]}}}$. In addition, a keyfob holder must possess a pedometer or an accelerometer for a precise self-gait observation, which nowadays is a common feature in smartphones. In fact, the cost and size for these utilities are consumer friendly in a sense that can also be integrated with something as small as a wireless keyfob.

\section{Adaptation with Existing Commitment Schemes}\label{sect:peddepi}
In this section, we illustrate further adaptations of our proposed solution based on initial knowledge pairing, i.e., proactive verification, and anthropomorphic feature, i.e., reactive verification. Interestingly, we choose two protocols, namely, Schnorr's identity-based commitment scheme ($\Gamma$) and Pedersen commitment scheme ($\Pi$) based on the \textit{commitment-before-knowledge} and \textit{Computational Diffie-Hellman (CDH)} assumptions.

\vspace{0.2cm}
\noindent
	\begin{figure}[h]
		\begin{framed}
			\begin{minipage}{6.4in}
				\footnotesize
				\noindent
				\textbf{The message flow}:
				
				\textit{Common inputs} are ($\wp,q,g,G$) \\
				\textit{Keygen(pk,sk)}: $sk=a {\in}_R F_{\wp}$ and $pk=g^a=A$ \\
				The protocol $\Gamma(\mathit{kf},v)$ steps for a $\mathit{kf}(a,A)$ and a verifier $v(A)$ are given below:
				
				\begin{center}
					\begin{tabular}{l l l}
						($\mathit{kf} \rightarrow v$) M1: & compute $X=g^x \in G$ such that $x \in_{R} F_{\wp}$ & \\
						($\mathit{kf} \leftarrow v$) M2: & send challenge $\Bbbk$ such that $\Bbbk {\in}_R F_{\wp}$ & \\
						($\mathit{kf} \rightarrow v$) M3: & reply response $\rho$ such that $\rho=x+a*{\Bbbk}$  & \\
					\end{tabular}
				\end{center}
				\noindent
				\textbf{Schnorr's protocol adaptation}:
				\vspace{0.1cm}
				\begin{enumerate}
					\item \textit{Commit:} Initially, a prover $\mathit{kf}$ commits a secret value $x$ by sending $X=g^x \in G$ to a verifier $v$. Next, the hashed string $H{(ed)}_{\mathit{kf}}$ (as derived using event data history $ed$) is concatenated with secret exponent $g^x$, i.e., $H(ed)||X$. In this case, secret exponent $x$ is not hidden but released secretly by just relying on the hardness assumption of $DL$ problem. Therefore, presumably, if an adversary is able to solve $DL$ then secret exponent part $x$ of the commitment would not remain secure and overall security might fallback on $H(ed)$ verification only.
					\item \textit{Challenge:} The recipient vehicle $v$ verifies the hashed string based on event history, and also retrieve the concatenated value $X$ as a public version of the secret commitment $x$. After receiving a public value $X$, verifier choose a random challenge $\Bbbk$ and send to the prover in order to seek for a valid paired response from an authentic prover $\mathit{kf}$.
					\item \textit{Open:} Consequently, prover $\mathit{kf}$ generate a paired response as $\rho=x+a*{\Bbbk}$. Interestingly, prover reveals the commitment value $x$ in a subliminal way. Obviously, the verification remains valid as long as $g^{\rho}==XA^{\Bbbk}$ satisfies on verifier side. Furthermore, there exist multiple variations to regular Schnorr like identification scheme. For example, one possible way~\cite{keser} is to replace $\rho$ with $\hat{\rho}$, i.e., $\hat{g}^{x+a*{\Bbbk}}$ where $\hat{g}=g^{\mathcal{H}(X|\Bbbk)}$.
				\end{enumerate}
			\end{minipage}
		\end{framed}
		\caption{Adaptation with the Schnorr's commitment protocol.}\label{fig:schen}
	\end{figure}
\vspace{0.2cm}

As the protocol, $\Gamma$ by Schnorr~\cite{beryyinluv} and protocol $\Pi$ by Chaum and Pedersen~\cite{chour} both are based on Diffie-hellman (DH) key exchange. Therefore, wlog, we assume that corresponding computations are done within a group $G$ $=$ $\langle g \rangle$ of prime order $q$, where CDH assumption holds.

\begin{definition}\label{def:military}
	Let $\langle g \rangle$ be a cyclic group $G$ generated by an element $g$ of an order $q$. There is no efficient probabilistic algorithm $\mathcal{A}_{CDH}$ that given $(g,g^a,g^b)$ produces $g^{ab}$, where $a$, $b$ are chosen at random from the group $G$.
\end{definition}

The CDH assumption satisfies that the computation of a discrete logarithm function $DL$ on public values $(g,g^a,g^b)$ is hard~\cite{deborah} within the cyclic group $G$.

\smallskip\noindent\textbf{Schnorr's scheme~\cite{beryyinluv}} The Schnorr's identification scheme $\Gamma$ is a natural extension of commitment-before-knowledge based paradigm as detailed in Figure~\ref{fig:schen}. According to protocol $\Gamma$, the prover chooses a secret DH exponent $x$ and releases a public value $X$ for the corresponding verifier. Consequently, the verifier returns a challenge $\Bbbk$ for the prover. Now, the prover generates a combined response $\rho$ such that it is computationally hard to compute $\rho$ without possessing the knowledge of $x$ and $\Bbbk$, i.e., $\rho=x+a*{\Bbbk}$ where $a$ is the long-term secret key paired with long-term public key $A$ presumably known to verifier. Ultimately, the verifier validates the received response $\rho$ as $g^{\rho}=XA^{\Bbbk}$ if correctly paired with the challenge $\Bbbk$ and public values $X$ or not.

\smallskip\noindent\textbf{Pedersen commitment~\cite{chour}.} Here, we provide the review of commitment scheme $\Pi$ as detailed in Figure~\ref{fig:ped}. According to the Pedersen's scheme, the system generates public values ($q,g,h,G$), where $g$ and $h$ are two generators such that $g=h^x$, for a secret key $x {\in}_R G$. Moreover, in order to commit a message $m$, prover knows ($x,q,g,h,G$) and the verifier knows ($q,g,h,G$), where $z=m^x$ is observed as a signature along with the proof that $log_gh=log_mz$.

The adapted solution further increases the security strength of the proposed scheme during second and third round of the handshake. It must be noticed that our scheme can incorporate other challenge-response protocols as well. Therefore, the proposed solution is open to further application dependent improvements via security integration.

\vspace{0.2cm}
\noindent
	\begin{figure}[h]
		\begin{framed}
			\begin{minipage}{6.4in}
				\footnotesize
				\noindent
				\textbf{The message flow}:
				
				Prover $\mathit{kf}$ knows $(x,y) {\in}_R F_{\wp}$ such that public key $X=(g^x*h^y)$. \\
				\textit{Common inputs} are ($\wp,q,g,h,G$) \\
				\vspace{-0.2cm}
				\begin{center}
					\begin{tabular}{l l l}
						($\mathit{kf} \rightarrow v$) M1: & compute commitment $A=g^s_1$ and $B=h^s_2$  & \\
						($\mathit{kf} \leftarrow v$) M2: & send challenge $\Bbbk {\in}_R F_{\wp}$  & \\
						($\mathit{kf} \rightarrow v$) M3: & reply response $\rho_1=s_1+{\Bbbk}*x$ and $\rho_2=s_2+{\Bbbk}*y$ & \\
					\end{tabular}
				\end{center}
				\noindent
				\textbf{Pedersen's protocol adaptation}:
				\vspace{0.1cm}
				\begin{enumerate}
					\item \textit{Commit:} Initially, prover $\mathit{kf}$ wants to commit a message $m \in F_{\wp}$. The prover chooses a random value $(s_1,s_2) \in F_{\wp}$, where $F_{\wp}$ is a finite field with elements (${0,1,...,\wp-1}$) and then compute the commitment $C=A*B$, where $A=g^s_1$ and $B=h^s_2$ in order to share with the verifier $v$, i.e., $H(ed)_{\mathit{kf}}||(c)$.
					\item \textit{Challenge:} The verifier $v$ computes the hashed string based on event data history and verify locally. After the verifications $v$ returns a challenge $\Bbbk {\in}_R F_{\wp}$ to prover.
					\item \textit{Open:} The verifier $v$ accepts the proof of commitment if $g^{\rho_1}*h^{\rho_2}=XC^{\Bbbk}$.
				\end{enumerate}
			\end{minipage}
		\end{framed}
		\caption{Adaptation with the Pedersen's commitment protocol.}\label{fig:ped}
	\end{figure}
\vspace{0.2cm}

\section{Security Analysis}
\label{sect:crpt}
In this section, we present a security analysis based on the following Assumptions~\ref{asm:van} and~\ref{asm:get}. Assumption~\ref{asm:van} is related to adversary model or adversarial behavior per se, i.e., a restriction on \textit{store-and-forward} syndrome (Condition~\ref{con:syn}) at an adversary. While, Assumption~\ref{asm:get}, in particular, specifies the significance of synchronized (at least within the pre-determined threshold) observation windows at the static vehicle and the corresponding keyfob holder, respectively. Next, we provide Condition~\ref{con:syn} and the related assumptions as below:

\begin{condition1}\label{con:syn}
	According to a store-and-forward syndrome, an adversary, while actively listening to the radio signals from an authentic access grantee (not closer to the static vehicle), ``intentionally'' delays the radio signal relaying so as to adapt the locomotion pattern of an authentic access grantee.
\end{condition1}

Condition~\ref{con:syn} presents a syndrome that an adversary might gain a fake access. However, the Assumption~\ref{asm:van} presents that an adversary cannot exercise the store-and-forward syndrome inconsistently, since that would yield a noticeable variation in any two consecutive propagation delay, i.e., with an adversary and without an adversary. Obviously, an adversary has to relay the communication consistently so as to remain undetectable. However, an adversary cannot relay the signals consistently because that would manipulate the locomotion pattern ($\mathit{vel}^v_{\mathit{kf}},\mathit{disp}^v_{\mathit{kf}}$) of keyfob holder as observed by the vehicle.

\begin{assumption}\label{asm:van}
	An adversary $\mathcal{A}$ cannot utilize the store-and-forward syndrome either consistently or inconsistently, such that: (\textit{i}) $(t^{\mathit{kf}}_{\mathit{cur}}-t^{\mathit{kf}}_{\mathit{send}})>(t^{v}_{\mathit{receive}}-t^{v}_{\mathit{send}})+t_{\epsilon}$ `or' (\textit{ii}) $t_{\mathit{relay}}$ is constant, respectively.
\end{assumption}\B

Assumption~\ref{asm:get} presents the similar observation at vehicle and keyfob, simultaneously. The deviation beyond a threshold $t_{\epsilon}$ would raise an alarm of relayed communication. It might result in false negatives (in which case handshake must be re-initialized) rather false positives.

\begin{assumption}\label{asm:get}
	\textit{The gait observation at the static vehicle and the dynamic keyfob holder should not differ beyond a threshold delay $t_{\epsilon}$, i.e. (${t_{\mathit{travel}}+t_{\epsilon}}\simeq{t_{\mathit{travel}}}$), where (${t_{\epsilon}}<{t_{\mathit{relay}}}$).}
\end{assumption}\B

Let us assume that an adversary $\mathcal{A}$ relays the messages between the static vehicle and corresponding keyfob while $t_{\mathit{relay}}$ is independent of $t_{\mathit{travel}}$. The adversary relay the messages successively for $n$ rounds and has the advantage as ${Adv(\mathcal{A})=Pr({1/2}^{n})}$. 
Now, considering the advantage that adversary $\mathcal{A}^{t_{\mathit{relay}}}$ can relay the messages only within the time bounds ($t_{\mathit{travel}}+t_{\mathit{relay}}$). Thus, the advantage  ${Adv(\mathcal{A}^{t_{\mathit{relay}}})=(\mathcal{A})*Pr({t_{\mathit{relay}}}>{t_{\epsilon}})}$ or
$${=[(1/2)*Pr({t_{\mathit{relay}}}>{t_{\epsilon}})*Pr({2t_{\mathit{relay}}}>{t_{\epsilon}})*({3t_{\mathit{relay}}}>{t_{\epsilon}})]}$$

The probability of time ($t_{\mathit{relay}}$) bound adversary winning the advantage is directly proportional to the difference in ($t_{\mathit{relay}},t_{\epsilon}$). Assuming that for every succeeding round the difference would accumulate and thus, ($Pr({2t_{\mathit{relay}}}>{t_{\epsilon}})>0$) in the second round. Thereby, on average with every succeeding round, there is at least as many chances to win as $Pr(1/2)$. 
Evidently, the proposed approach with average time-bound ($t_{\mathit{relay}}$) for every succeeding round would restrict the advantage to ${Adv(\mathcal{A})}-{Adv(\mathcal{A}^{t_{\mathit{relay}}})}$. 

The proposition~\ref{prp:perp} illustrates the hard and rock situation for the adversary, i.e., an adversary will either be present consistently throughout multiple rounds or will be present for selective rounds only. In both cases, the presence of any relayed communication can be detected.

\begin{proposition}~\label{prp:perp}
	\textit{If the adversary is present in only one of the first two rounds it must be revealed to both vehicle and the keyfob, by the end of step 3. Conversely, if the adversary is present uniformly throughout the course of gait observation it would be accountable via differing gait observations.}
\end{proposition}

The proposition~\ref{prp:tor} emphasize that any observable dynamics (location vs distance) will distinguish between a regular communication and an affected communication.

\begin{proposition}\label{prp:tor}
	\textit{No adversary can relay the messages such that both gait observations at the vehicle $\mathit{vel}^v_{\mathit{kf}}$ and the keyfob $\mathit{vel}_{\mathit{kf}}$ are equivalent or within the bounds of $t_{\epsilon}$. The presence of an adversary will yield ($\mathit{vel}^v_{\mathit{kf}}>{\mathit{vel}_{\mathit{kf}}}$).}
\end{proposition}


Theorem~\ref{them:firs} proves an impossibility to extract the pre-determined state of the vehicle, i.e., an exact series of events; given that an adversary is allowed to observe the vehicle movements during the last itinerary of that vehicle. Precisely, the proof is based on the number of geometric trajectories that a vehicle could possibly drive through during last itinerary, i.e., state based on the most recent driving pattern.

\begin{theorem}\label{them:firs}
	The series of event occurrence while following a continuous geometric trajectory is predictable with only a negligible probability.
\end{theorem}

\begin{proof} Let us consider a subset $L$ that encompasses straight-line trajectories in a set, say, $S$ denoted as $y=mx+b$, where ($x,y$) are coordinates, $m$ is the slope, and $b$ is the intercept with the $y$-axis. Now considering the closed interval [$w, z$] in which the line expands across the Euclidean plane. Interestingly, according to the Cantor's un-countability proof, the number of points in a line segment of the trajectory is the same as in an entire line (one-dimension) and as the number of points in a $n$-dimensional space. Considering, (\textit{i}) cardinality of the straight line trajectories, (\textit{ii}) cardinality of the curved trajectories and (\textit{iii}) cardinality of the event series over continuous time series.
	
	\begin{lemma} \label{lem:one}
		The cardinality of straight line trajectories $L$ in set $S$ is the same as the cardinality of continuum.
	\end{lemma}
	
	\begin{proof} According to (\textit{i}) the cardinality of straight line trajectories $L$ spanning over the road segment is same as the cardinality of continuum $\mathcal{C}$. Let us first compute the cardinality of all possible line segments through a particular coordinate point with a different slope. Thus, all those line segments that have a unique slope from a right open interval $[0,2\pi)$ and fixed coordinates ($x,y$) should be countable. Clearly, the cardinality of all these line trajectories is same as the $\mathcal{C}$ because the unique slope $m$ for each line trajectory is parameterized by reals. Second, let us consider the other remaining subset of straight line trajectories that have a fixed slope $m$ and expand for all possible ($x,y$) intercepts. Analogous to the Cantor hypotheses the line segment would have as many points as any $n$-dimensional space, therefore, there will be as many straight line trajectories as $\mathcal{C}$. Consequently, there will be $\mathcal{C}*\mathcal{C}$ lines trajectories in total.
	\end{proof}

	\begin{lemma}\label{lem:two}
		The cardinality of curved trajectories $C$ in $S$ is the same as the cardinality of straight line trajectories $L$.
	\end{lemma}
	
	\begin{proof} A one to one onto function can be assumed for a realistic trajectory followed by any vehicle. According to (\textit{ii}) the cardinality of curved trajectories, $C$ can be derived from the cardinality of straight line trajectories as shown in Lemma 1. However, in spite of considering the cardinality of multi-degree polynomials, we precisely consider the case of a parabolic trajectory and then additively generalize it to the multi-degree curves. Since adding the cardinality of multi-degree curves on top of the cardinality of parabolic trajectories will only populate the subset $C$, hence, will increase the cardinality. The derivation of any segment of a parabolic trajectory analogs the derivation of at least one pair of straight lines tangent to each other, which at most can be as many as pairs of real numbers. Furthermore, the slope of these tangent lines is uniquely determined at any point on the curve. Since the tangent slope is uniquely parameterized from a right open interval $[0,2\pi)$ as a real number it is straight that the cardinality of parabolic trajectories will have the cardinality of straight lines $L$. Therefore, the cardinality of subset $C$ is at least as big as the cardinality of $L$ that is $\mathcal{C}$.
	\end{proof}

	\begin{lemma}\label{lem:three}
		The cardinality of the discrete event set over a continuous time function is also continuous hence continuum.
	\end{lemma}
	
	\begin{proof} Let us consider the probability of any discrete event occurrence on a continuous time series. The continuous time function can be imagined as a collection of overlapping intervals $I=\{I_1, I_2,...I_n\}$, i.e. $n$ independent events can occur in $i^{th}$ interval $I_i$. Since the timeline is of the order of real numbers the probability of event occurrence is $\mathcal{C}/n$.
	\end{proof}

	\begin{lemma}\label{lem:four}
		An adversary has a negligible advantage in retrieving the exact series of tuples from the last driving interval.
	\end{lemma}
	
	\begin{proof} An adversary successfully retrieves the tuple $\{\mathit{gt}, \mathit{time}, e\}$ if correctly predicts the geometric trajectory from $S$ and independent events $e$ over continuous time series. Therefore, assuming a probabilistic polynomial time adversary, an attack would only have a negligible advantage as $n/\mathcal{C}$. In particular, an adversary needs to have $n/2^{\aleph_0}$ number of attempts for a successful random guess in polynomial time.
	\end{proof}
	
	Consequently, an adversary will have a negligible advantage in retrieving the internal event state of the vehicle with respect to the geometric trajectory over different time intervals, hence, Theorem~\ref{them:firs} is straight to articulate as below:
	
	The series of event occurrence during adjacent drives by the same vehicle through the same distance will always differ with a non-negligible probability. Let us consider the events ($e$) span over a 2d Euclidean plane 
	with x-axis and y-axis. A continuous geometric trajectory ($gt$) followed by the vehicle expands over the x-axis while still allowing the movement across y-axis such as during lane changing or curved maneuvering. An adversary $\mathcal{A}$ that successfully predicts the recent internal state of the vehicle. In order to predict the replicated state preserved by the key-fob adversary must have predicted a series of tuples ($gt, t, e$). The geometric trajectory expands over a set $S$ of straight lines and curves, i.e., $S=\{L, C\}$. It must be noticed that any curve inside the set $S$ is an umbrella term for a parabola or any higher degree polynomial. However, we prove that the parabolic trajectory covered by any vehicle between two different time stamps ($t_i, t_j$) is unpredictable with a negligible probability. Therefore, the series of events (which might have been observed correctly from a relative distance) followed during that trajectory is also unpredictable with a negligible probability.
\end{proof}

%


The basic peripheral device authentication scheme does not assume any external channel bandwidth during the authentication except the computation requirements, i.e., 3 \texttt{exponent} and 1 \texttt{hash digest} computation at each party\footnote{Note that in case of adaptation with the Schnorr and Pedersen schemes the exponent computation cost will increase by 1 and 2 at each party, repectively.}. While the vehicle to keyfob authentication scheme requires an additional communication overhead with a maximum packet size that can hold the hash digest in each communication round. In addition, other parameters such as nounce, identity, time, displacement are scalar values and does not require much payload size. Moreover, different IoT appropriate communication protocols such as Wifi, Wifi (with WEP, WPA, WPA2), Zigbee, Z-wave, Bluetooth, 6LoWPAN, varies within the 64 Byte (minimum) to 2324 Byte (maximum) range. Therefore, we asume that the proposed scheme is feasible for deployment in IoT settings with the available communication protocols.



\section{Related Work}\label{sect:prev} 
Remote keyless entry via transponder integrated with the ignition key or using immobilizer has been used actively~\cite{toward}.
A solution based on a distance bounding protocol and a verifiable multi-lateration scheme has been considered in~\cite{surgon,toward}. An immobilizer is used to avoid the vehicle movement even if an unauthorized access is gained. The immobilizer-based physical lock is useful to secure the physical periphery but not the digital periphery. In~\cite{tech} a coalition attack scenario has been solved, however, it differs in access control countermeasures required for a parked vehicle scenario. Therefore, our solution eliminates the need for specialized physically unclonable function (PUF) integration to the vehicle and the keyfob. Moreover, another work regarding the wireless device pairing~\cite{shake} allows a user identification based on the common movement. Also, a gait based authentication scheme is given in~\cite{gpat} for identifying that the same user is operating over two devices simultaneously.
The distance bounding~\cite{dbound,terror} provides a distance upper bound based on the round trip time measurement through multiple rounds of CRP exchange. However, our work is based on two-way CRP handshake to continuously measure the gait (independently) over a pre-defined interval, which is further verified in a third round. Apparently, EDR installation does not impose any safety challenges as the sole purpose of installation is the event recording. However, considering the data confidentiality and integrity issues the event log must be secure from any external penetration~\cite{atlit}. According to the existing work~\cite{meets} EDR recording can be secured by using a public-key cryptosystem; whereas any long-term records are not accessible without the consent of the vehicle owner. However, the latest records stored in a short-term storage must be accessible to the authorities while allowing privacy via time release functions (TRF). 

\section{Conclusion}\label{sect:fut}

The proposed solution provides a secure digital periphery for the autonomous
vehicles. The scheme is resilient against the \textit{distance fraud} and \textit{mafia fraud},
however, as a part of future direction, we aim to provide an additional protection against
the \textit{terrorist fraud} as well. According to the terrorist fraud, an authentic prover assists
with the adversary (by handing over secret component) to impersonate in front of the
verifier. Therefore, this attack scenario is practically feasible even with the biometric authentication
because the original secret holder can authorize himself, and, let the service be accessible to others who do not possess secret biometric credentials. The countermeasures would
require a more sophisticated identity verification method such as a ticket granting authority
(issuing tickets for service access) or the definability feature. 




\end{document}